\documentclass[%
 reprint,
 amsmath,amssymb,
 aps,
]{revtex4-1}

\usepackage{graphicx}% Include figure files
\usepackage{dcolumn}% Align table columns on decimal point
\usepackage{bm}% bold math
\usepackage{hyperref}% add hypertext capabilities
\usepackage[caption=false]{subfig}
\usepackage{siunitx}
\usepackage{amsthm}

\newtheorem{theorem}{Theorem}
\newtheorem{definition}{Definition}
\newtheorem{corollary}{Corollary}[theorem]

\begin{document}

\preprint{APS/123-QED}

\title{Evaluating Performance of Photon-Number-Resolving Detectors}% Force line breaks with \\

\author{Mattias J\"{o}nsson}
\email{matjon4@kth.se}
% \altaffiliation[Also at ]{Physics Department, XYZ University.}
\author{Gunnar Bj\"{o}rk}%
\email{gbjork@kth.se}
\affiliation{%
 Department of Applied Physics, KTH Royal Institute of Technology\\
 AlbaNova Univerity Center, SE 106 91 Stockholm, Sweden
}%


%\affiliation{
%\affiliation{
%\affiliation{%


\date{\today}% It is always \today, today,

\begin{abstract}
We analyze the performance of photon-number-resolving (PNR) detectors and introduce a figure of merit for the accuracy of such detectors. This figure of merit is the (worst-case) probability that the photon-number-resolving detector  correctly predicts the input photon number. Simulations of various PNR detectors based on multiplexed single-photon ``click detectors'' is performed. We conclude that the required quantum efficiency is very high in order to achieve even moderate (up to a handful) photon resolution, we derive the required quantum efficiency as a function of the the maximal photon number one wants to resolve, and we show that the number of click detectors required grows quadratically with the maximal number of photons resolvable.
\end{abstract}

\maketitle

%\tableofcontents

\section{\label{sec:introduction}INTRODUCTION}
Photon-number-resolving (PNR) detectors have been shown to be useful in various optical applications, such as linear optical quantum computing \cite{Knill2001}, quantum key exchange \cite{PhysRevA.98.012333}, entanglement distribution \cite{Pan2017}, photon-counting laser radars \cite{Huang2014}, X-ray astronomy \cite{Holland1999}, evaluation of single photon sources \cite{Hadfield:05}, and elementary particle detection \cite{Haba2008}. Due to the wide applicability of these detectors many different schemes have been proposed; both inherent PNR detectors \cite{Kim1999, Irwin1995, PhysRevA.71.061803, Waks2003, Rosenberg2005a, Rosenberg2005b, Lita2008, Lita2009, Fukuda2011, Cahall2017} and various schemes for multiplexing single-photon detectors as to construct what is called a segmented PNR detector \cite{Banaszek2003, Fitch2003, Achilles2003, Achilles2004, Eraerds2007, Jiang2007, Divochiy2008, Guerrieri2010, Afek2009, Natarajan2013, Mattioli2015, Nehra2017}. However, so far only limited work has been dedicated to systematically and realistically evaluate the actual performance of segmented PNR detectors \cite{Achilles2003,Fitch2003,Natarajan2013, Nehra2017}.

Figures of merit used to evaluate the performance in previous work has generally been input signal dependent, which makes it difficult to determine how the detectors will perform in applications with varying input signals. Furthermore, the figures of merit used for single photon detectors, such as quantum efficiency, can not be directly applied to PNR detectors since they can have elaborate internal detection mechanisms and therefore behave differently for different input signals. 

In this paper we introduce a figure of merit, the PNR quality, to evaluate the accuracy of PNR detectors. We show that this figure is a natural generalization of the quantum efficiency. The figure is input signal independent for a selected set of input signals and any uncertainty within this set. However, generalizing the set to any input, the figure becomes completely input independent. We also show how the PNR quality can be used to set an upper, resolvable input photon-number for PNR detectors such that the number prediction for any input not exceeding this number can be trusted to within a specified ``confidence'' level.

We simulate three different segmented PNR detector schemes, a spatial array, a temporal array constructed using fiber couplers and a loop-multiplexed detector, and evaluate their PNR performance. Our analysis is much in the spirit of \cite{Nehra2017,Kruse2017} and qualitatively our results agree but the assumptions and details differ. We show quantitatively how the quantum efficiency of the single-photon detector(s) limit the number of photons resolvable with each PNR detector type, and that the needed number of detector elements grows quadratically with the number of photons one desires to resolve. We also analyze quantitatively how the dark-count probability affects the segmented PNR-detector quality.

\section{\label{sec:detector-quality}PNR QUALITY}
An ideal PNR detector gives an output signal indicating the input number of photons independent of the input signal used. However, in practice all detectors have non-ideal characteristics, so there exists some upper input number $n$ which is the largest number of photons that the detector is capable of resolving with reasonable certainty. To incorporate such non-ideal properties in a model, we introduce the $n$-detector which is a PNR detector capable of detecting up to $n$ photons, with probabilities $P_{k, m}$ that the output is $S_O = k$, given that the input was $S_I = m$ photons. The probability to receive an output larger than $n$ is zero and thereby the $n$-detector is effectively classifying the input signal into the $n + 1$ output classes $0, 1, \dots, n-1$, and $ \geq n$ photons. It is possible to transform any PNR detector into a $n$-detector by mapping the output $S_O \mapsto \min \{ S_O, n \}$.

Quite naturally we say that the output from a PNR detector was desired if the input was classified into the correct class, which for an $n$-detector corresponds to
\begin{equation}
    S_{O, \text{ desired}} = 
    \begin{cases}
        m & \text{if } S_I = m \leq n\\
        n & \text{otherwise}
    \end{cases}.
\end{equation}
Given the desired output, we can define a figure of merit to be the probability that the $n$-detector's output is desired, although this figure will in general be dependent on the input signal and we therefore need to select a specific, completely general input signal set to remove the dependence.

To eventually arrive at an input-signal independent figure of merit we first introduce the PNR quality, which is the smallest probability that the $n$-detector's output is the desired one on a set of possible input photon-number distributions $\mathcal{F}$. Formally we define the PNR quality as
\begin{equation}
\begin{split}
    Q_n(\mathcal{F}) \equiv \inf_{p \in \mathcal{F}} \Bigg( \sum_{m \leq n} P_{m, m} p(m) + \sum_{m > n} P_{n, m} p(m) \Bigg),
\end{split}
\end{equation}
where the set $\mathcal{F}$ is a subset of all probability distributions on $\mathbb{N}$. Consequently, it holds that if the input signal distribution is in $\mathcal{F}$ then the probability for the output to be the desired one is at least $Q_n(\mathcal{F})$. The set $\mathcal{F}$ should therefore be selected so it includes all input signals that the detector is expected to handle, for example any Poissonian photon-number distribution. 

The choice of input photon-number distribution set that the detector is expected to handle and the maximal number of photons that can be resolved by the detector both affect the PNR quality of the detector. For the former it holds trivially that a reduction of the distribution set from $\mathcal{F}$ to $\mathcal{F}' \subset \mathcal{F}$ can help to improve the PNR quality of the detector. Hence, quite naturally, there is a trade-off between allowing a large input set and having a high guaranteed probability to get the correct output. For the latter it holds generally that (see appendix \ref{sec:proofs})
\begin{equation}
    Q_n(\mathcal{F}) \leq Q_{n-1}(\mathcal{F}) \quad \forall \, \mathcal{F}, n > 0,
\end{equation}
so by reducing the maximal number of photons that can be resolved we may improve the PNR quality of the detector.

Finding the desired output probability infinum on a function set can be difficult if the number of free parameters is large, or if there is uncertainty which signal is used as input. For example, assume that there are two possible input signals $p_1$ and $p_2$, but that there is an uncertainty which of the two which will be sent. The set of possible input distributions is then not $\mathcal{F} = \{ p_1, p_2 \}$, but rather
\begin{equation}
    \mathcal{F}' = \{ a p_1 + b p_2 \mid a + b = 1 \land a, b \geq 0 \}.
\end{equation}
However, the uncertainty does not increase the complexity of the optimization problem since it holds that
\begin{equation}
    Q_n(\mathcal{F}') = Q_n(\mathcal{F}),
\end{equation}
so it is therefore possible to simplify the minimization problem to only consider $\mathcal{F}$. Generally it holds that (see appendix \ref{sec:proofs})
\begin{equation}
    Q_n(\mathcal{A}) = Q_n(\mathcal{B})
\end{equation}
if it is possible to write all elements in $\mathcal{A}$ as a linear combination of elements in $\mathcal{B}$ where the coefficients are non-negative.

Let us consider a special case of the PNR quality when $\mathcal{F}$ is the set of all probability distributions on $\mathbb{N}$. We denote this special case with $Q_n$ and we can show that (see appendix \ref{sec:proofs})
\begin{equation}
    Q_n = \min \bigg\{ \min_{m \leq n} P_{m, m}, \; \inf_{m > n} P_{n, m} \bigg\}.
    \label{eq:quality-full-set}
\end{equation}
Hence, the most difficult input signals to resolve are the distributions with $\SI{100}{\percent}$ probability for some photon number, i.e., any Fock state. Furthermore, for a single-photon detector with quantum efficiency $\eta$ and dark count rate $r_d = 0$ the quality of the detector is $Q_1 = \eta$, which makes the signal $p(k) = \delta_{k, 1}$ the most difficult signal to resolve. The quality for the full probability set could be thought of as a generalization of the quantum efficiency for single-photon detectors.

A reasonable requirement on a PNR detector is that it should outperform guessing the outcome. If a detector has $Q_n \geq 0.5$ then it has better-than-guessing quality for any signal with unknown probabilities consisting of 2 or more outcomes. Therefore we will in the following often use this value to define for what photon number $n$ specific PNR detectors can reasonably be said to resolve $0, 1, \ldots, n-1$ or $n$ or more, input photons. 

\section{\label{sec:detector-simulation}DETECTOR SIMULATION}
\begin{figure}[t]
    \centering
    \includegraphics[width=0.45\linewidth]{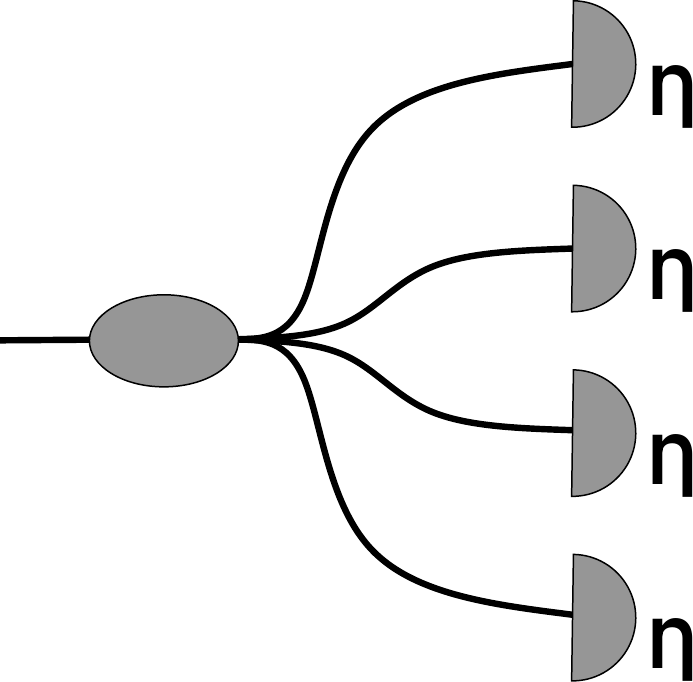}
    \caption{Schematic image of a spatial array consisting of 4 detector elements with quantum efficiency $\eta$. An input signal equally distributed over the four detector elements, e.g., with a 1-to-4 fiber coupler.}
    \label{fig:space-schematic}
\end{figure}

% Result for the spatial array
\begin{figure*}[t]
    \centering
    \subfloat[\label{fig:quality-space-8}]{%
       \includegraphics[width=.494\linewidth]{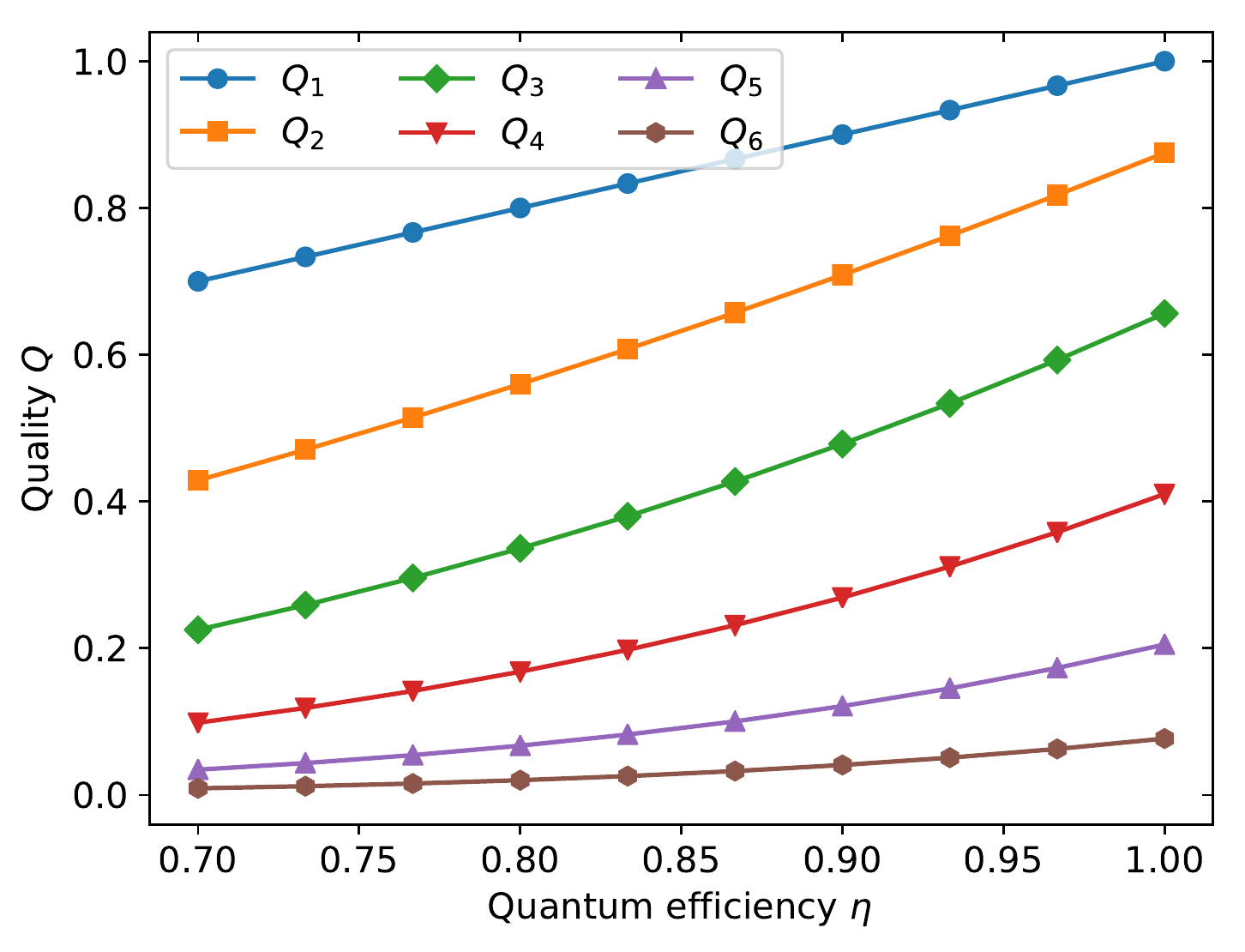}
    }
    \subfloat[\label{fig:poisson-quality-space-8}]{%
        \includegraphics[width=.494\linewidth]{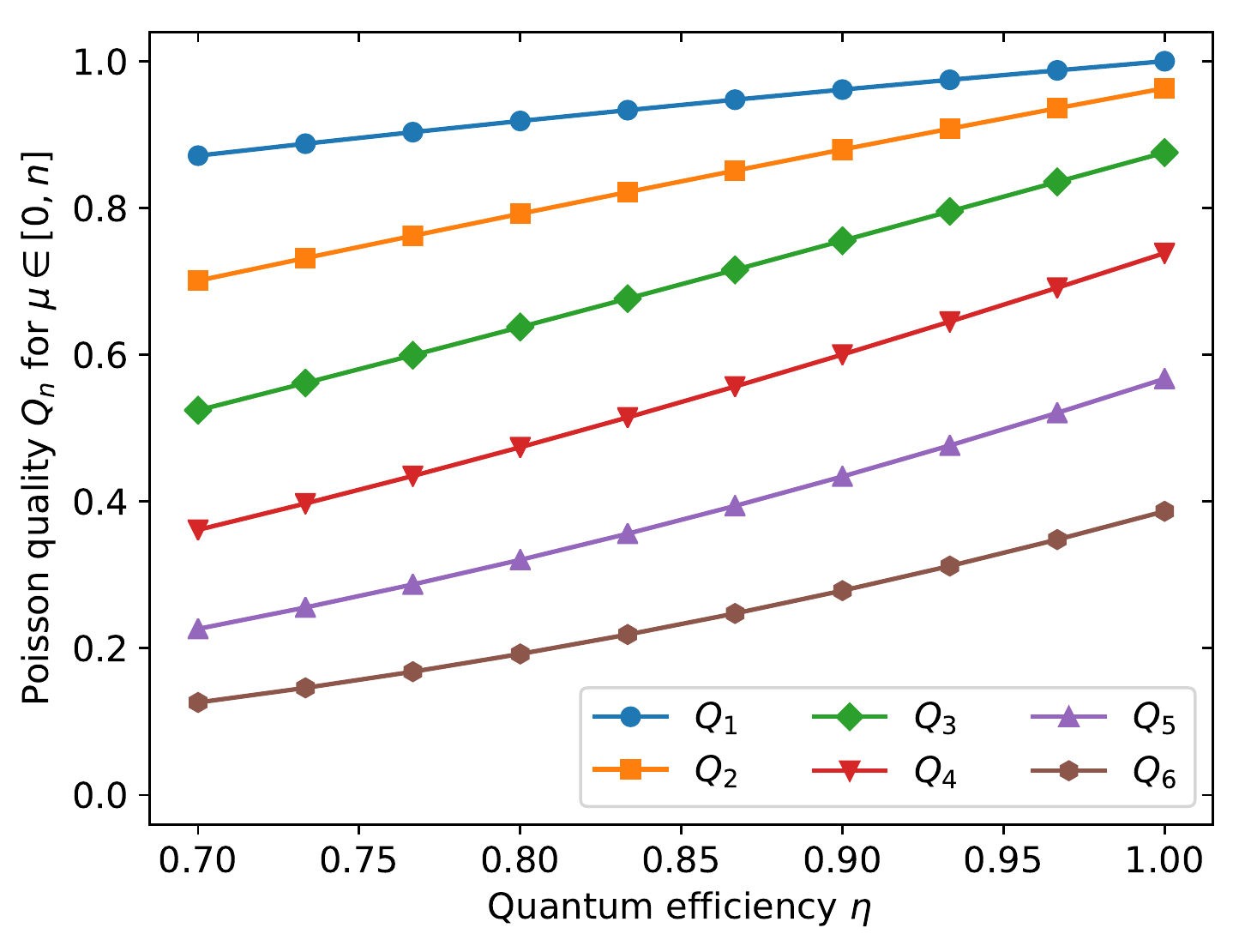}
    }
    
    \subfloat[\label{fig:quality-space-32}]{%
       \includegraphics[width=.494\linewidth]{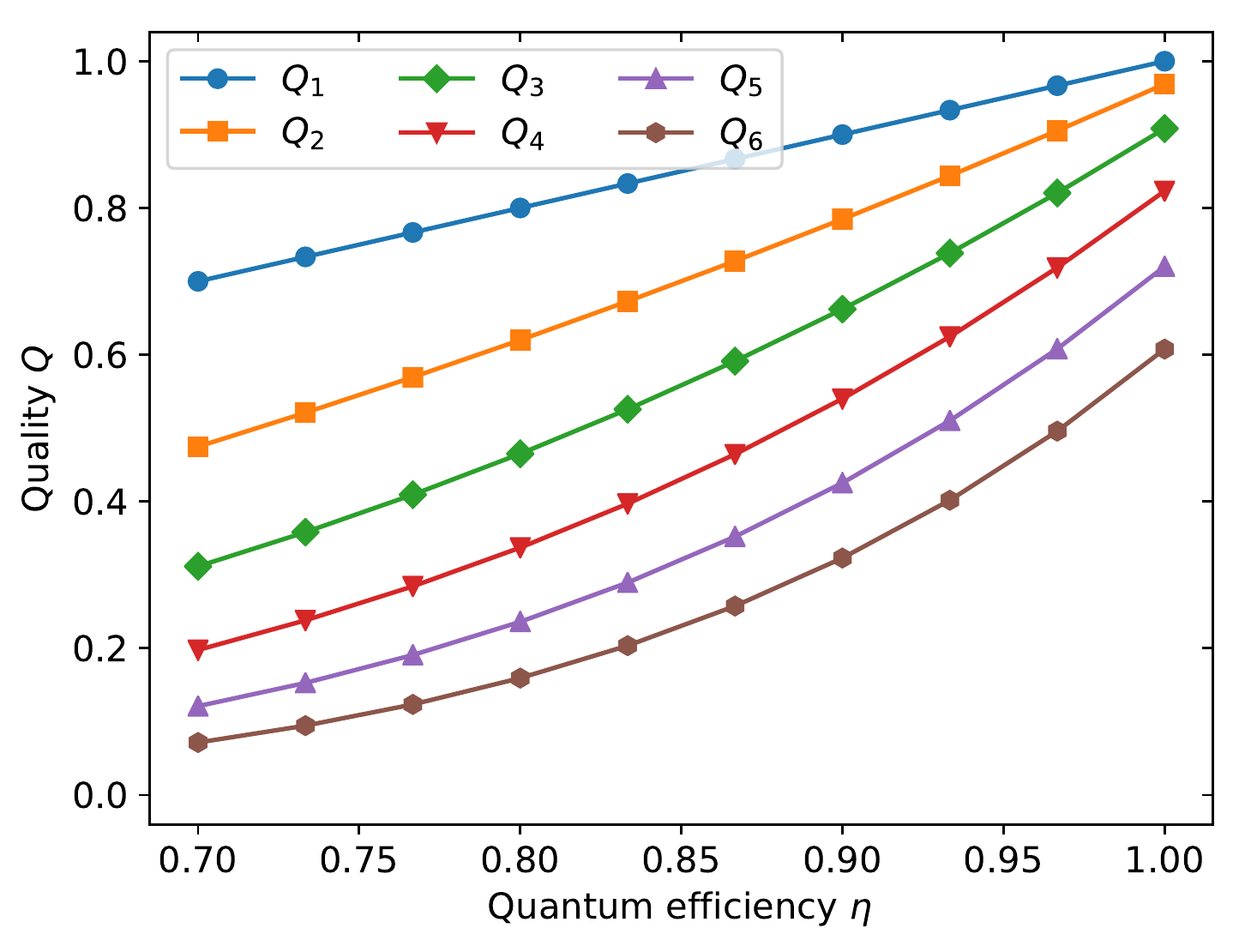}
    }
    \subfloat[\label{fig:poisson_quality_space_32}]{%
        \includegraphics[width=.494\linewidth]{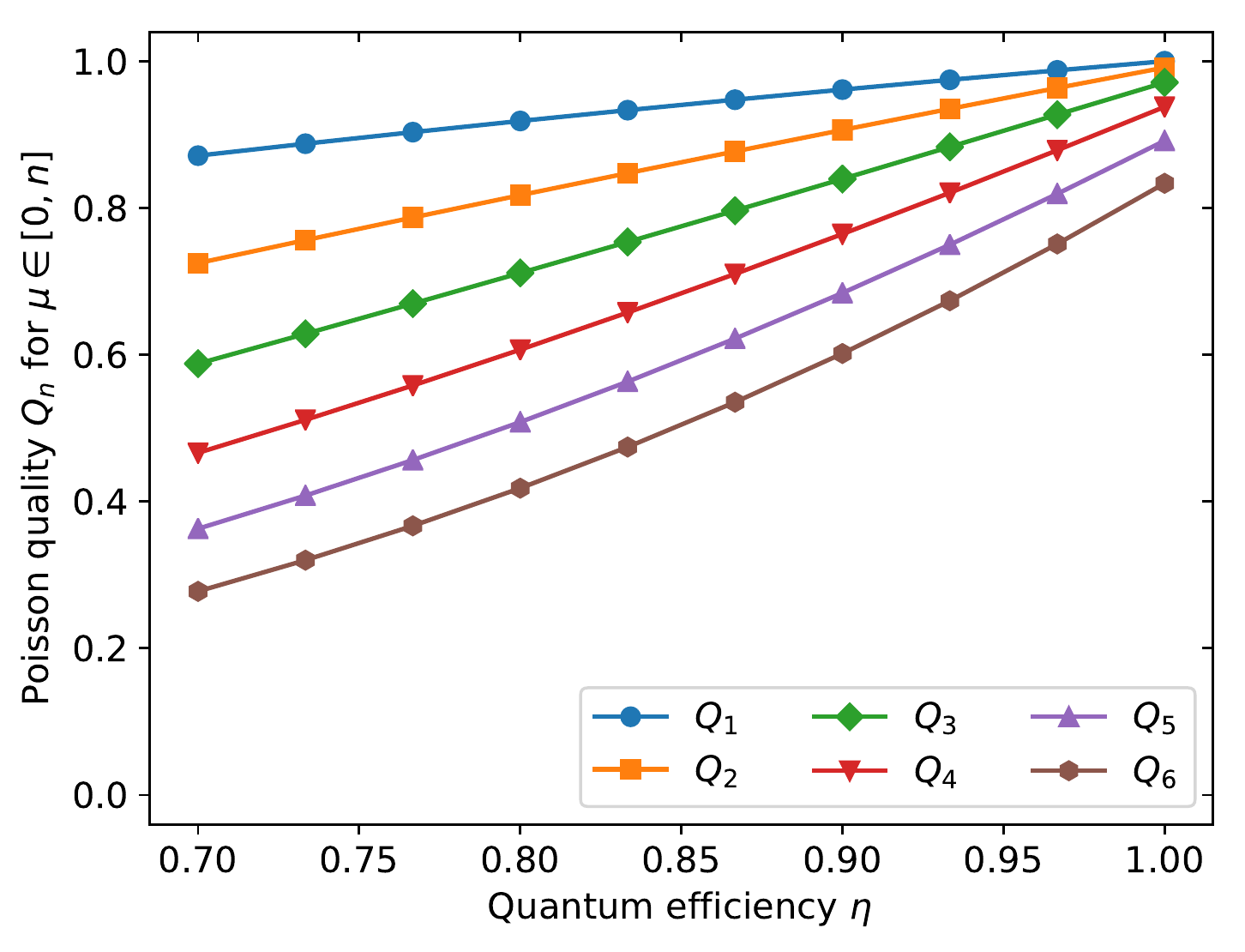}
    }
    \caption{\label{fig:space-quality}Simulation result of the PNR quality for spatial arrays with negligible dark count rates. In (a) and (c), the PNR quality using the set of all probability distributions is presented for an array consisting of 8 and 32 elements, respectively. In (b) and (d), the PNR quality using the set of Poisson distributions with mean $\mu \leq n$ is presented for an array consisting of 8 and 32 elements, respectively. Both array sizes show improvement in the quality when restricted to the set of Poisson distributions. However, reaching a quality larger than $0.5$ for many photons still requires a very high quantum efficiency.}
\end{figure*}

In this section we present the PNR quality obtained from simulations of three different ``segmented'' detectors, a spatial array, exemplified in \cite{Eraerds2007,Jiang2007,Divochiy2008,Guerrieri2010,Mattioli2015}, a temporal array exemplified in \cite{Achilles2003,Fitch2003,Natarajan2013,Kruse2017}, and a loop-multiplexed detector such as in \cite{Banaszek2003,Haderka2004}. We limit ourselves to two different function sets, the set of all probability distributions on $\mathbb{N}$ and the set of all Poisson distributions
\begin{equation}
    \mathcal{F} = \bigg\{ f: k \in \mathbb{N} \mapsto \frac{\mu^k e^{-\mu}}{k!} \mid \mu \in [0, n] \bigg\},
\end{equation}
where we have limited the mean $\mu$ to be at most $n$ for a $n$-detector. We run this simulation for different values of $n$ to determine what quantum efficiency is required to reliably ($Q_n \geq 0.5$) detect $n$ photons.

In the simulations we model the three multiplexed detectors as devices that distribute photons on click elements. Thus the individual detector are assumed to only distinguish between zero and one or more photons, while the segmented detector is able to use the combined result from the individual detector elements to detect more than one photon. To get the probabilities $P_{k, m}$ for different outcomes of the multiplexed detector we sum over all possible distributions on the detector elements and all possible detection outcomes when the photons hit the single-photon detector elements. The exact distribution of photons over the segments is implementation dependent and is presented in the subsections below, however the detector element model is shared among the considered PNR detectors. We assume that if a click-detector has quantum efficiency $\eta$, dark count probability $p_d$ and $m$ photons hit the detector then the probability for the detector to click is
\begin{equation}
    \Pr(\text{click} \mid m) = 1 - (1 - p_d) (1 - \eta)^m.
\end{equation}
That is, we have assumed that the dark count probability is independent of the photon detection probability.

\subsection{\label{sec:spatial-arra-simulation}Spatial array}
The spatial array consists of $M$ click detector elements and the photons are distributed equiprobably over the elements as seen in Fig. \ref{fig:space-schematic}. A prototype implementation is a uniformly illuminated array of single-photon avalanche photodiodes \cite{Jiang2007, Guerrieri2010} or an array of superconducting nano-wire detectors \cite{Divochiy2008, Mattioli2015}. The outcomes when $m$ photons is used as input can be represented as a matrix $\vec{x} \in \mathbb{N}^M$, where element $k$ is the number of photons that hit detector $k$ and $||\vec{x}||_1 = m$. The corresponding probability for each outcome is given by the multinomial distribution
\begin{equation}
    \Pr (\{ x_1, x_2, ... \}) = \frac{M!}{M^m \prod_{i = 1}^M x_i!}.
\end{equation}

\begin{figure}[t]
    \centering
    \includegraphics[width=1\linewidth]{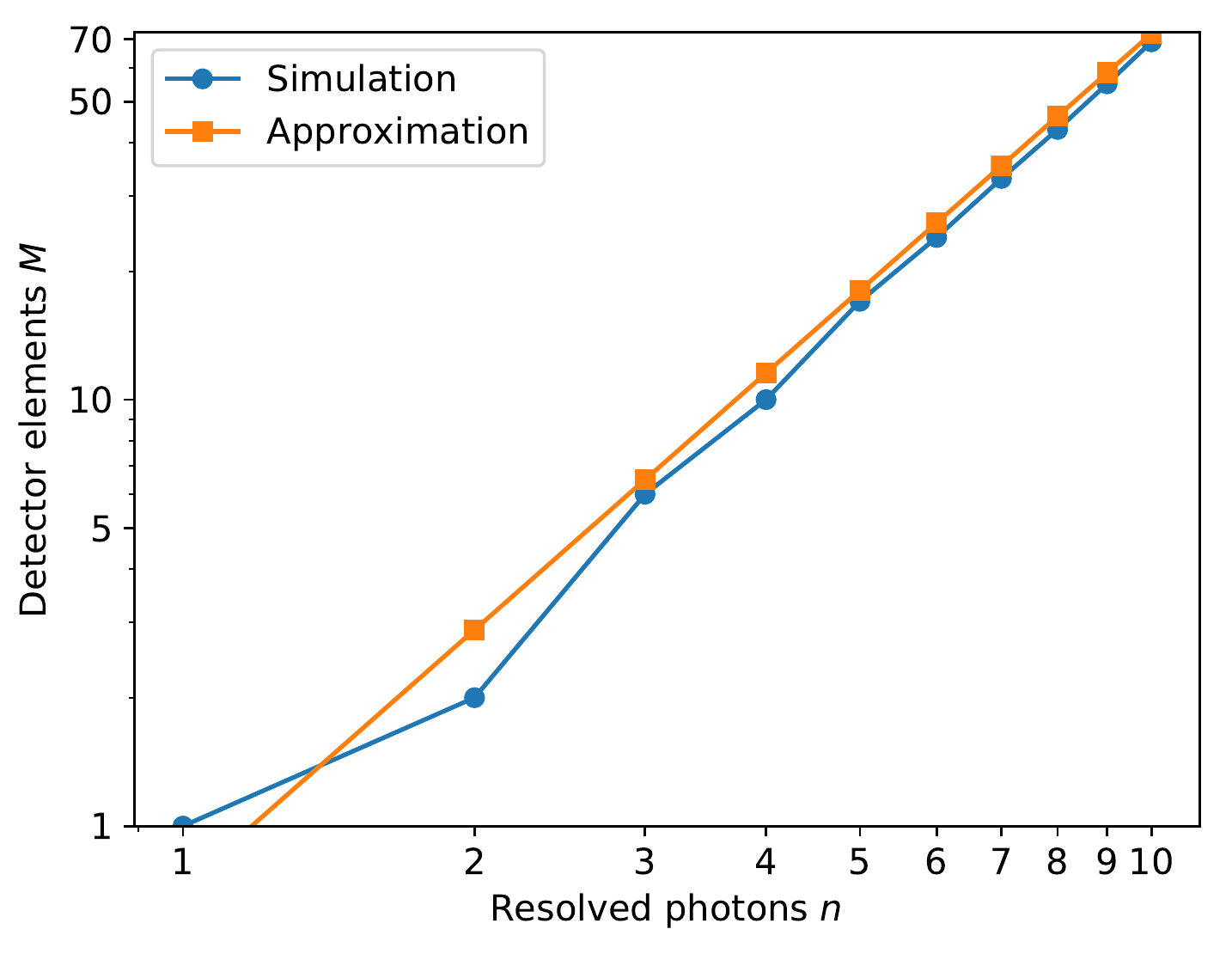}
    \caption{The number of detector elements $M$ needed get a PNR quality $Q_n \geq 0.5$ in a spatial array with quantum efficiency $\eta = 1$ and negligible dark counts as a function of the number of resolved photons. The simulated result is in agreement with the approximation in Eq. (\ref{eq:space-detector-scaling}) for sufficiently large $M$.}
    \label{fig:quality-scaling}
\end{figure}

The probabilities $P_{k, m}$ for the spatial array are in general difficult to compute analytically, however it is possible to show that 
\begin{equation}
    P_{m, m} = \frac{M!}{M^m (M - m)!} \eta^m.
    \label{eq:space-array-all-detected}
\end{equation}
This gives us an upper bound for the PNR quality $Q_n$ and it holds that $Q_n \leq \eta^n$. Consequently, for a given PNR quality the maximal number of photons that can be resolved is bounded by $n \leq \ln{Q_n} / \ln{\eta}$, independent of the number detector elements in the array. Furthermore, if we assume that $M \gg n$ we can show, using Eq. (\ref{eq:space-array-all-detected}) and Stirling's formula, that
\begin{equation}
    M \geq \frac{1}{2} \frac{n^2}{n \ln{\eta} - \ln{Q_n}} + \mathcal{O} \bigg( \frac{n}{M} + M \ln{M}  \bigg),
    \label{eq:space-detector-scaling}
\end{equation}
if $n < \ln{Q_n} / \ln{\eta}$. In Fig. \ref{fig:quality-scaling} it is shown that the number of detector elements $M$ needed is well approximated by Eq. (\ref{eq:space-detector-scaling}) for sufficiently large $M$. This implies that for a fixed PNR quality the number of detector elements grows quadratically with the number of photons that the PNR detector can resolve according to our quality criterion. It was noted already in \cite{Sperling2012} that the number of detector elements must be much larger the number of photons one wants to resolve. In Fig. 2 of \cite{Sperling2012} one sees that for 9 photons the minimum $M$ is somewhat less than 100 which is in agreement with our Fig. \ref{fig:quality-scaling}. To resolve 5 photons $\geq 20$ elements are needed, and to resolve 10 photons one needs $\geq 72$ detector elements. Hence, to resolve many photons, many detector elements, all of them having a high quantum efficiency, are required. As we shall see, it is very challenging to build PNR detectors using the spatial array scheme.

\begin{figure}[t]
    \centering
    \includegraphics[width=1\linewidth]{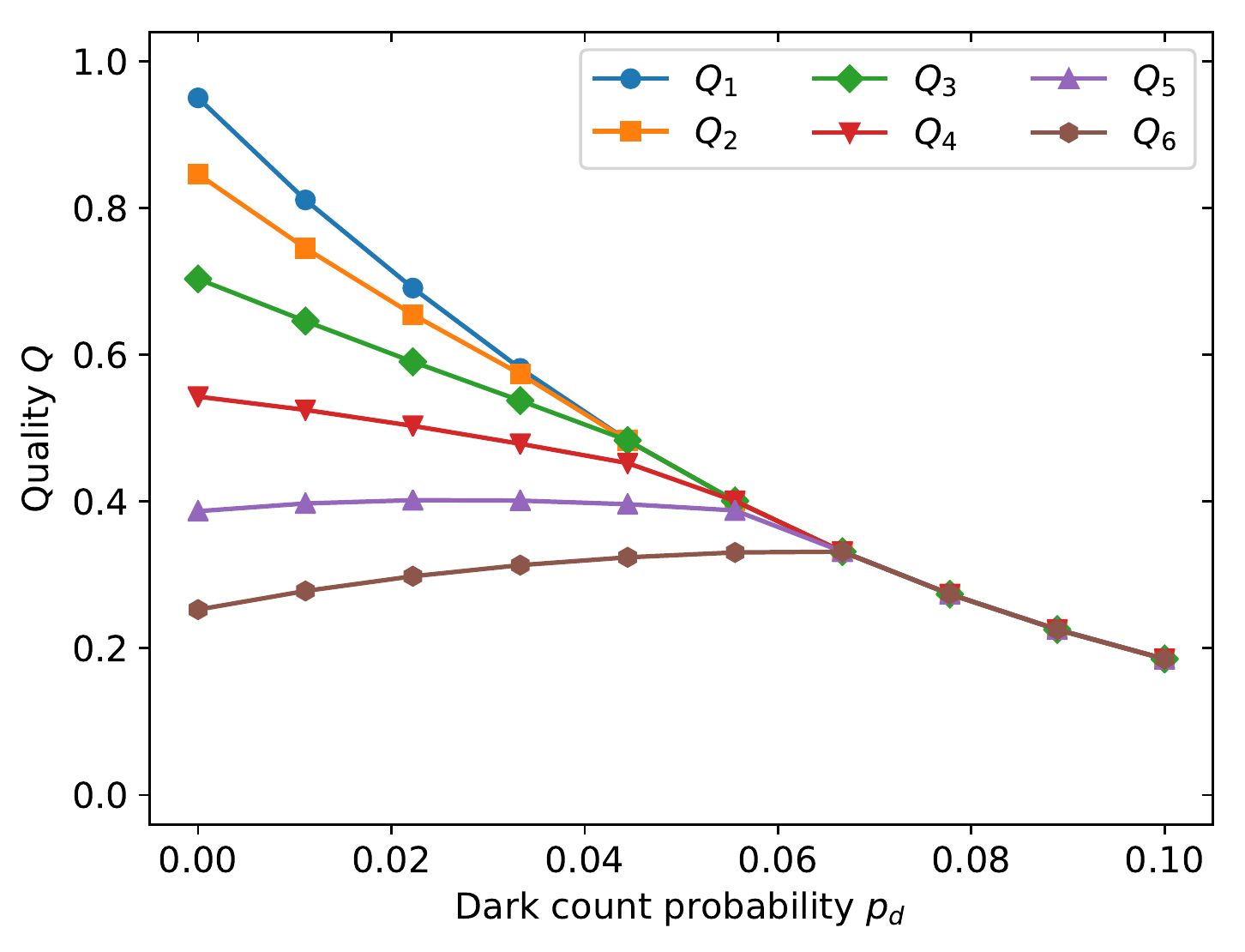}
    \caption{PNR quality for a 16 element array with quantum efficiency $\eta = 0.95$ as a function of the dark count probability of each detector element. The general trend shows that the PNR quality decreases with increased dark counts, however in some instances it is possible to have an increasing PNR quality.}
    \label{fig:space-quality-with-dark-counts}
\end{figure}

In Fig. \ref{fig:space-quality} the simulated PNR quality is presented for spatial arrays consisting of 8 and 32 detector elements where dark counts have been neglected. As expected there is a significant improvement for both arrays when the distribution set is restricted to Poisson distributions, although the requirement on the quantum efficiency is still very high in order to resolve many photons. Comparing an eight-segment detector in Fig. \ref{fig:space-quality} \subref{fig:quality-space-8} with a 32-segment detector in \ref{fig:space-quality} \subref{fig:quality-space-32} we notice that the input photon number for which $Q \geq 0.8$ doubles for $\eta = 1$ which is what was predicted by Eq. (\ref{eq:space-detector-scaling}).

 Using the $Q_n \geq 0.5$ requirement to evaluate the spatial arrays yields the result that an ideal 8 detector element array (as seen in Fig. \ref{fig:space-quality} \subref{fig:quality-space-8}) is only capable of resolving 3 photons if all possible input signals are allowed. If restricted to Poisson distributions (as seen in figure \ref{fig:space-quality} \subref{fig:poisson-quality-space-8}) such a detector is able to detect up to 5 photons, although in both cases the requirement on the quantum efficiency is close to unity ($\eta \geq 0.92$ and $\eta \geq 0.96$, respectively).

In Fig. \ref{fig:space-quality-with-dark-counts} it is shown how the PNR quality is affected by the dark count probability for a fixed quantum efficiency. As expected, the general trend is that the PNR quality quickly decreases as the dark count probability increases, yet $Q_6$ violates this trend in the region $p_d \in [0, 0.06]$ where it is increasing. Hence, the probability to resolve the signal that minimized the desired output must have been increased by the added dark counts. Loosely speaking, the dark counts compensate for the non-detected incident photons for high photon numbers. However, at the same time the probability to resolve other signals decrease ($Q_1$ to $Q_4$ all decrease with $p_d$) so quite intuitively it is therefore not beneficial for the overall PNR detector performance to add dark counts for the purpose of increasing the quality for large $n$.

\subsection{\label{sec:temporal-array-simulation}Temporal array}
The temporal array consists of two single-photon detectors and a series of fiber couplers that split the signal equally between different fiber paths. Such setups, each with three couplers, have been reported in \cite{Achilles2003,Fitch2003, Natarajan2013}. The lengths of the paths should be chosen both so that the detectors have time to recover between photons taking different paths and also so that a photons taking different paths never combine and interfere in a subsequent coupler. One scheme that can fulfill these two criteria is presented in Fig. \ref{fig:temporal-schematic}, where $l$ and $L$ are introduced. The length $l$ is chosen to be some arbitrary length, while $L$ is chosen to be the shortest distance between light pulses needed for the click detectors to recover. The length of the top and the bottom fiber after coupler $k$ is selected to be $2^{k-1} L + l$ and $l$, respectively \cite{Fitch2003}. In \cite{Kruse2017} it was shown that for the number of detector elements we are are considering ($< 100$), fiber dispersion effects can be neglected.

\begin{figure}[t]
    \centering
    \includegraphics[width=1\linewidth]{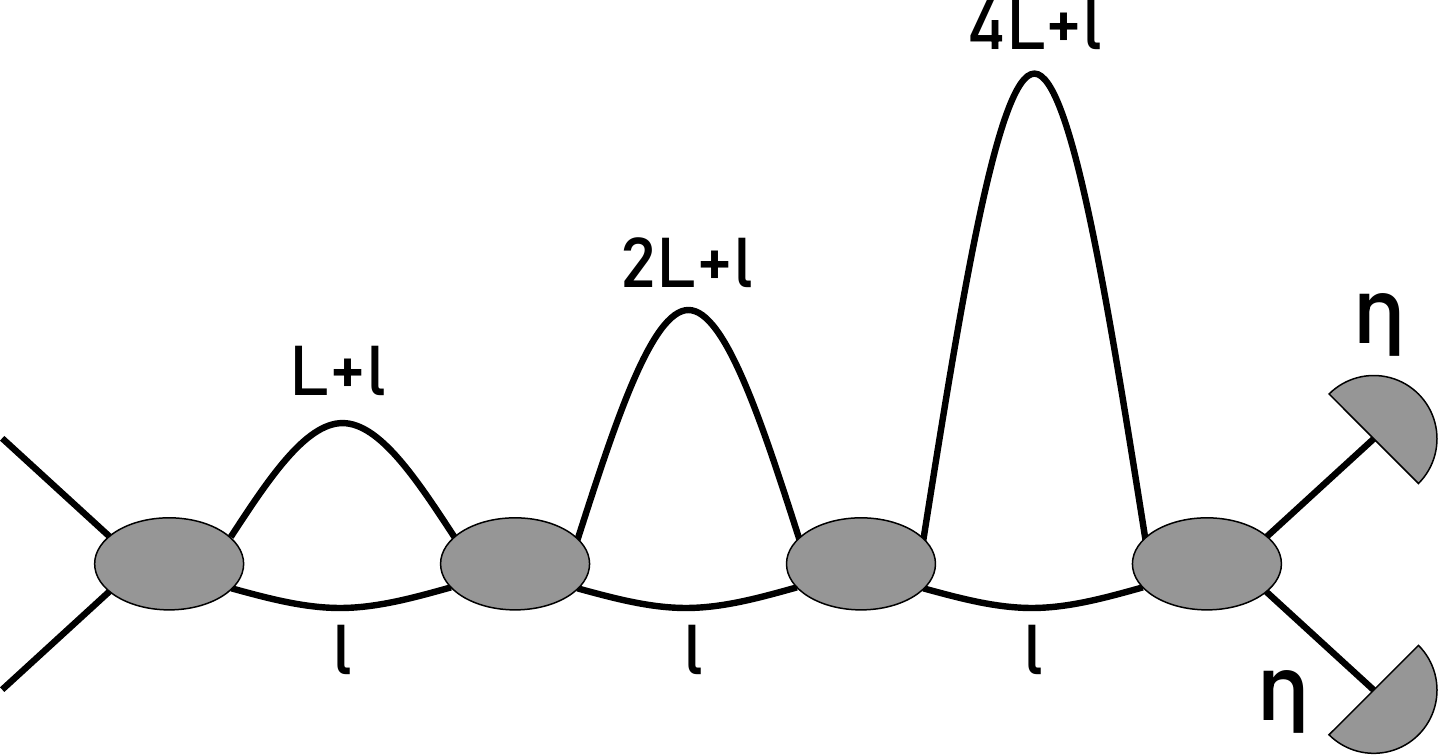}
    \caption{Schematic image of a temporally segmented detector. The fibers between the 50:50 couplers are chosen so the detectors have time to recover between pulses and so a split signal can not recombine. Between coupler $k$ and $k+1$ the lower fiber has a short (but arbitrary) length $l$ and the upper fiber has length $2^{k-1} L + l$, where $L$ corresponds the smallest length rendering two subsequent pulses that the detector can resolve.}
    \label{fig:temporal-schematic}
\end{figure}

The setup in Fig. \ref{fig:temporal-schematic} produces eight equidistantly spaced pulses with equal amplitude at each detector from a single input pulse. Thus the scheme mimics, by temporal splitting, a 16-detector element spatial array. This scheme thus results in a detector with an accuracy that is mathematically equivalent to the accuracy of a uniformly illuminated spatial array and the PNR quality is therefore described by Fig. \ref{fig:space-quality}. However, by introducing the couplers in front of the detector we increase the effective recovery time of the temporally segmented detector by a factor $M/2$, where $M$ is the number of effective detector segments. Moreover, the quantum efficiency of the segmented detector is lowered by the linear losses in the fiber couplers. If the fiber couplers have an efficiency $\eta_c$ and the number of effective detector segments is $M$ then
\begin{equation}
    \eta_{\text{eff}} = \eta_c^{\log_2 M} \eta.
\end{equation}
Hence, in this configuration there exists a trade-off between having a large number of detector elements and having a high quantum efficiency. At some point, by increasing the number of couplers one will therefore decrease the PNR quality due to a drop in quantum efficiency.

\subsection{\label{sec:loop-multiplexed-simulation}Loop-multiplexed detector}

\begin{figure}[t]
    \centering
    \includegraphics[width=0.5\linewidth]{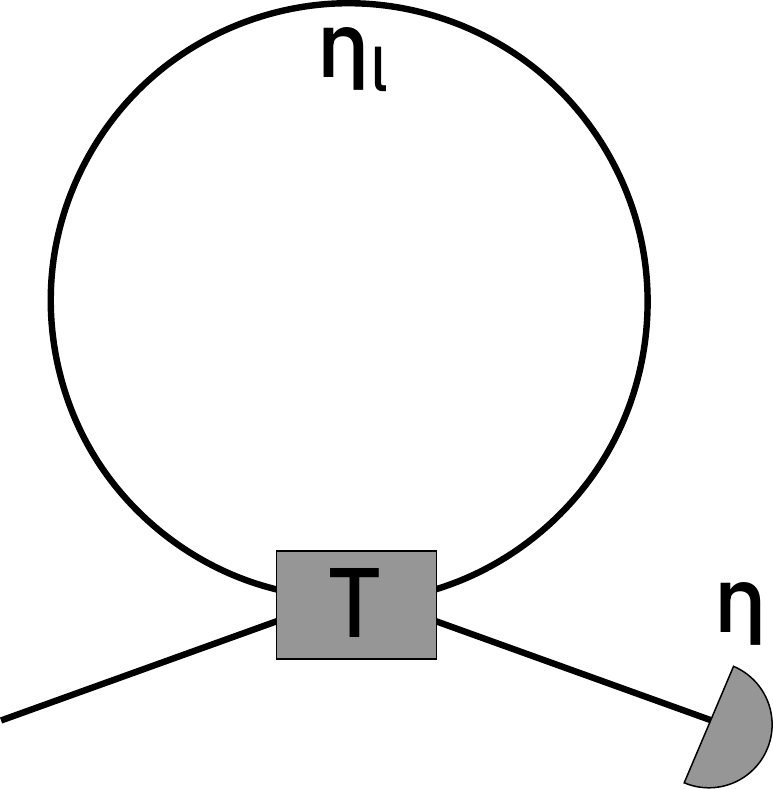}
    \caption{Schematic image of a loop-multiplexed detector. Incoming light enters the multi-mode fiber end to the left and enters a splitting circuit in the middle. The probability for photons to exit to the detector is $T$ independent of which of the two input fibers was used by the photon. The remaining probability is that the photon enters the loop, where it has a $\eta_l$ probability to survive per loop.}
    \label{fig:loop-schematic}
\end{figure}

% Result for the loop multiplexed detector
\begin{figure*}[ht]
    \centering
    \subfloat[\label{fig:loop-quality-32}]{%
       \includegraphics[width=.494\linewidth]{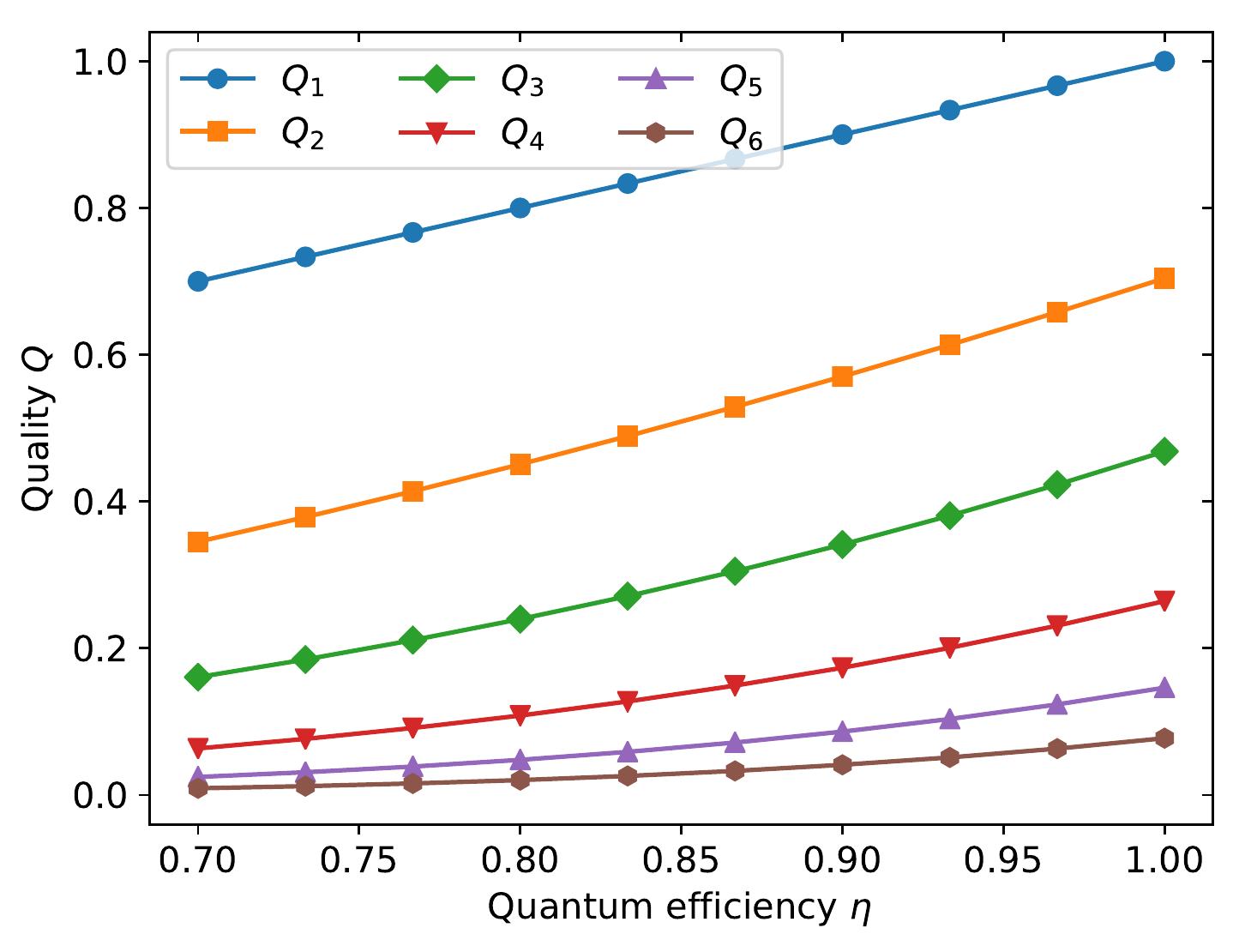}
    }
    \subfloat[\label{fig:loop-poisson-quality-32}]{%
        \includegraphics[width=.494\linewidth]{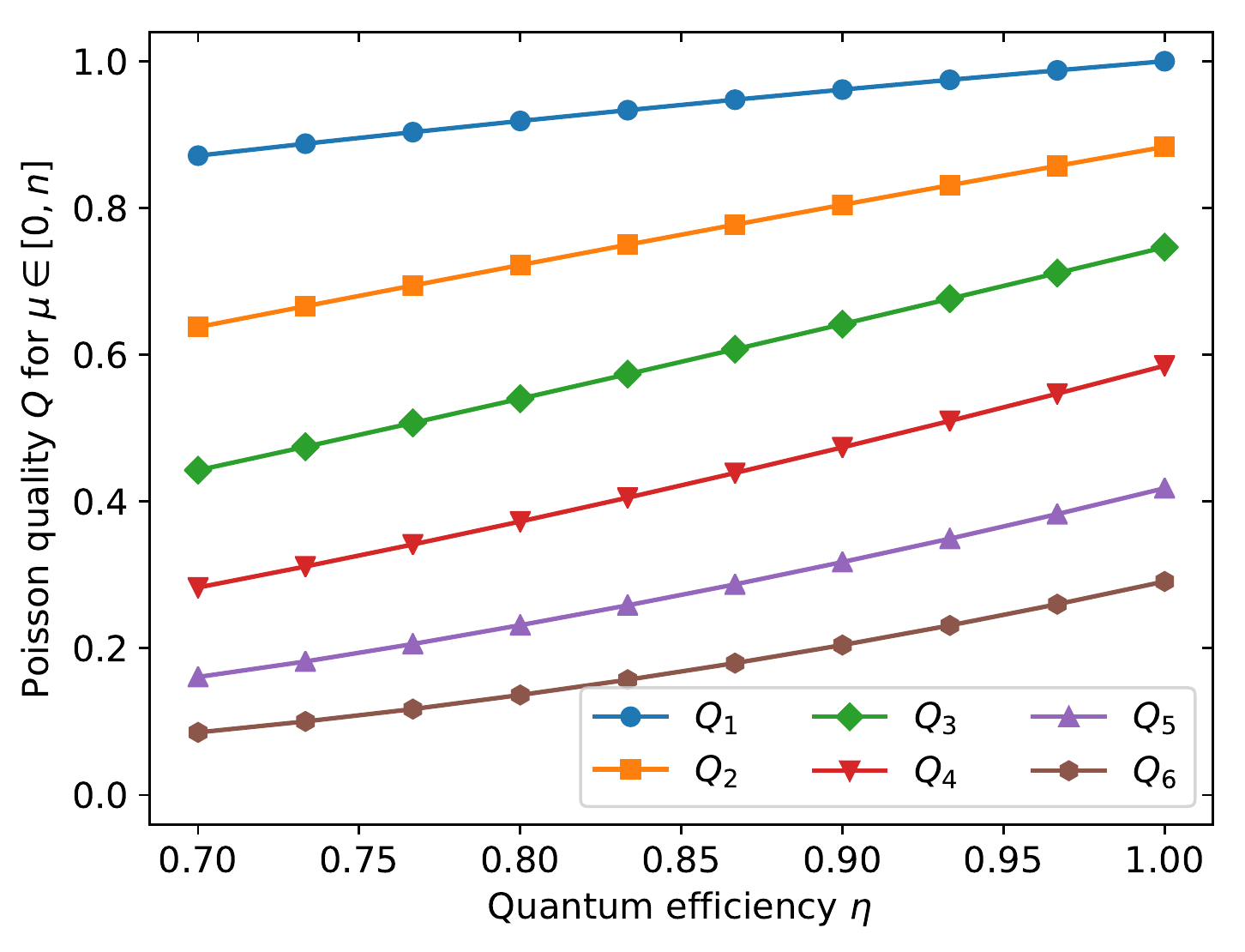}
    }
    \caption{\label{fig:loop-quality}PNR quality of a loop-multiplexed detector without dark counts. The probability that a photon survives the loop is $\eta_l = 0.97$ and at most $l = 32$ loops are allowed. (a) The PNR quality for all distributions is presented. Only two photons can be detected with a quality higher than $0.5$. (b) The PNR quality for all Poisson distributions with mean $\mu \leq n$. Compared to (a) the quality has improved significantly, yet it is still only possible to detect up to four photons even with a perfect detector.}
\end{figure*}
The loop-multiplexed detector consists of a multi-mode fiber loop with inherent mode scrambling, a multi-mode to single mode coupler circuit and a single photon detector as seen in Fig. \ref{fig:loop-schematic}. Properly designed, when a input enter the circuit each photon has a probability $T$ to exit the loop and hit the detector. The remaining photons loop back into the circuit and the process repeats from the beginning up to $l$ times. Hence, during each revolution in the loop a photon has probability $T$ to exit. It is not possible to build such a splitting circuit with a single mode fiber coupler since it would result in an exit probability $T$ at the first coupler passage and a probability $(1 - T)^2$ at the second passage. Consequently, using a single mode fiber coupler results in every photon having at least a $\SI{75}{\percent}$ probability to exit the loop during the first two passes of the coupler \cite{Haderka2004,Banaszek2003}. This makes it difficult to resolve multiple ($>2$) photons with any confidence using a click-detector in such a setup.

The probabilities for different outcomes in the loop-multiplexed detector can be described with a recursion relation
\begin{equation}
\begin{split}
    P^{(l)}_{k, m} &= \sum_{a, b} f_{\text{bin}}(t, a , m) f_{\text{bin}}(\eta_l, b , m - a) \times\\
    &\times \Big(\Pr(\text{click} \mid a) P^{(l - 1)}_{k - 1, b}\\
    &+ \Pr(\text{not click} \mid a) P^{(l - 1)}_{k, b} \Big),
\end{split}
\end{equation}
where $\eta_l$ is the probability that a photon survives the loop and can re-enter the circuit, $P^{(l)}_{k, m}$ is the probability to get an output $k$ given that the input is $m$ and at most $l$ loops are allowed after which the detector is turned off (or its output is ignored) and $f_{\text{bin}}(p, a, b)$ is the binomial probability to get $a$ successful outcomes given $b$ tries and a probability $p$.

Figure \ref{fig:loop-quality} shows the simulated result from the loop-multiplexed detector where we have assumed a maximal number of allowed loops $l = 32$. As seen in Fig. \ref{fig:loop-quality} \subref{fig:loop-quality-32}, for arbitrary inputs only two photons can be resolved with a PNR quality larger than $0.5$ and in Fig. \ref{fig:loop-quality} \subref{fig:loop-poisson-quality-32} only four photons can be resolved with a PNR quality of $0.5$ or more if the input photon distribution is limited to Poisson distributions. Consequently, this detector is outperformed by a temporal array with $32$ detector elements if the losses in the couplers are assumed to be equal to $\eta_l$.

\section{\label{sec:summar}SUMMARY}
In this paper we introduced a figure of merit that is useful to assess the resolution accuracy of PNR detectors. The figure can be made input-signal independent and is equal to the smallest probability that the detector gives the correct output.

Simulations of three different PNR detectors, implemented as multiplexed click-detectors show that the requirements on the quantum efficiency are very high in order to resolve a handful, or more, photons. With an eight segment detector, one cannot resolve more than three photons with better-than-guessing quality even with ideal click-detectors. Furthermore, the needed number of detector segments in an array grows quadratically with the number of photons resolved, so large arrays consisting of high quantum efficiency detector elements are imperative to resolve more than a few photons.

In the literature one encounters claims of PNR detectors that appears to contradict the limits we derive. Here it is important to differentiate between resolution of the input photon number and the resolution between the output signals corresponding to different numbers of absorbed photons. In \cite{Divochiy2008} array-PNR detectors are presented containing 4-6 segments. Even at unit quantum efficiency the input photon-number resolution ability of such detectors is limited to $n=3$, but at the reported quantum efficiency of 2 \%, not even single photon input resolution is reached. In contrast, in the (few) events where 1, 2, or 3 photons were absorbed by the array, the three different output signals are well resolved. Thus, the authors' claim about output signal resolution is correct. However, the output signals do not allow one to conclude much about the input photon number. In \cite{Mattioli2015} arrays with 4, 5, 12, and 24 elements were reported. Again, a 4 element array can only resolve, with accuracy, two photons, even at unit quantum efficiency. At the reported quantum efficiency of 0.17 \% (for the 12 element detector) none of the arrays have even single photon resolution although, again, the output signals corresponding to the absorption of $0, 1, \ldots 24$ photons are shown to be resolvable. Finally in \cite{Eraerds2007}, a 132 element detector is reported with a 16 \% efficiency. The authors make no claims about input photon number resolution but note that the signals corresponding to different number of absorbed photons are well resolved. According to our resolution criterion the detector is not even able to resolve single input photons from none.

With the technology available today only a few detector types, such as transition-edge detectors \cite{Lita2008,Fukuda2011} and superconducting nanowire detectors \cite{Marsili2013,EsmaeilZadeh2017} have quantum efficiencies high enough to give better-than-guessing quality for more than a few photons. Even so, most realizations of PNR detectors with spatial arrays of click-detectors have difficulty to increase the number of detector elements to what is required for better-than-guessing PNR capability. Therefore, temporal arrays seem to be the most reasonable option to implement segmented PNR detectors if the losses in the fiber couplers can be made small and if some temporal resolution can be sacrificed.

\section*{\label{sec:acknowledgments}Acknowledgments}
This work was supported by the Knut and Alice Wallenberg Foundation grant "Quantum Sensing", the Swedish Research Council (VR) through grant 621-2014-5410, and through its support of the Linn\ae us Excellance Center ADOPT.

The simulations were performed on resources provided by the Swedish National Infrastructure for Computing (SNIC) at PDC.

%\bibliography{references}

%

\appendix

\section{\label{sec:proofs}PROOFS}

\begin{theorem}\label{the:limiting-detectors}
Let $\mathcal{F}$ be a subset of all probability distributions on $\mathbb{N}$ and let $Q_n(\mathcal{F})$ be the PNR quality on the set for an $n$-detector. It then holds that
\begin{equation}
    Q_{n}(\mathcal{F}) \leq Q_{n - 1}(\mathcal{F}) \quad \forall \, \mathcal{F}, n > 0.
\end{equation}
\end{theorem}

\begin{proof}
Assume that $P_{k, m}$ and $P'_{k, m}$ are the conditional probabilities for the $n$-detector and the $(n - 1)$-detector, respectively. A $n$-detector is made into a $(n-1)$-detector by mapping the output $S_O \mapsto \min \{S_O, n-1 \}$ and hence $P'_{n-1, m} = P_{n-1, m} + P_{n, m}$. Furthermore, the desired output of the $n$-detector is
\begin{equation}
\begin{split}
    &\sum_{m \leq n} P_{m, m} p(m) + \sum_{m > n} P_{n, m} p(m)\\
    &=P_{n, n} p(n) + \sum_{m \leq n-1} P_{m, m} p(m) + \sum_{m > n} P_{n, m} p(m)\\
    &\leq(P_{n-1, n} + P_{n, n}) p(n) + \sum_{m \leq n-1} P_{m, m} p(m)\\
    &+ \sum_{m > n} (P_{n-1, m} + P_{n, m}) p(m)\\
    &= \sum_{m \leq n-1} P'_{m, m} p(m) + \sum_{m > n-1} P'_{n, m} p(m),
\end{split}
\end{equation}
where the expression in the last equality is the desired output of a $(n - 1)$-detector. It follows directly that $Q_{n}(\mathcal{F}) \leq Q_{n - 1}(\mathcal{F})$ if $n > 0$.
\end{proof}

\begin{definition}\label{def:extended-set}
Let $\mathcal{F}$ be a a subset of all probability distributions on $\mathbb{N}$. The extended distribution set contains $\mathcal{F}$ and all countable linear combinations of the elements in $\mathcal{F}$, so 
\begin{equation}
    \bar{\mathcal{F}} = \mathcal{F} \cup \bigg\{ \sum_{k} a_k p_k \mid p_k \in \mathcal{F}, \sum_k a_k = 1, a_k > 0 \bigg\}.
\end{equation}
\end{definition}

\begin{theorem}\label{the:extended-quality}
For any set $\mathcal{F}$ it holds that
\begin{equation}
    Q_n(\mathcal{F}) = Q_n(\bar{\mathcal{F}}).
\end{equation}
\end{theorem}

\begin{proof}
Introduce the function 
\begin{equation}
    f: p \in \bar{\mathcal{F}} \mapsto \sum_{m \in \mathbb{N}} R_{m, m} p(m),
\end{equation}
where the coefficients
\begin{equation}
    R_{m, m} = 
    \begin{cases}
        P_{m, m} & \text{if } m \leq n\\
        P_{n, m} & \text{otherwise}
    \end{cases}.
\end{equation}
It then holds by definition that $Q_n(\mathcal{A}) = \inf_{p \in \mathcal{A}} f(p)$ for $\mathcal{A} \subseteq \bar{\mathcal{F}}$.

We can assume that $\bar{\mathcal{F}} \backslash \mathcal{F} \neq \varnothing$, otherwise the theorem holds trivially. Take $p \in \bar{\mathcal{F}} \backslash \mathcal{F}$ and compute
\begin{equation}
\begin{split}
    f(p) &= \sum_{k \in \mathbb{N}} a_k f(p_k)\\
    &\geq \inf_{l \in \mathbb{N}} f(p_l) \sum_{k \in \mathbb{N}} a_k\\
    &= \inf_{l \in \mathbb{N}} f(p_l).
\end{split}
\end{equation}
By definition $p_l \in \mathcal{F} \; \forall l \implies \inf_{l \in \mathbb{N}} f(p_l) \geq Q_n(\mathcal{F})$ so $Q_n(\mathcal{F})$ is a lower bound to $f$ on $\bar{\mathcal{F}} \backslash \mathcal{F}$. It therefore follows that $Q_n(\bar{\mathcal{F}} \backslash \mathcal{F}) \geq Q_n(\mathcal{F})$ and therefore it must be true that $Q_n(\bar{\mathcal{F}}) = Q_n(\mathcal{F})$.
\end{proof}

\begin{corollary}\label{the:equal-extensions}
Let $\mathcal{A}$ and $\mathcal{B}$ be subsets of all probability distributions on $\mathbb{N}$. If $\bar{\mathcal{A}} = \bar{\mathcal{B}}$ then it holds that
\begin{equation}
    Q_n(\mathcal{A}) = Q_n(\mathcal{B}).
\end{equation}
\end{corollary}

\begin{proof}
Using theorem \ref{the:extended-quality} yields
\begin{equation}
    Q_n(\mathcal{A}) = Q_n(\bar{\mathcal{A}}) = Q_n(\bar{\mathcal{B}}) = Q_n(\mathcal{B}).
\end{equation}
\end{proof}

\begin{theorem}
The PNR quality for the set of all probability distributions on $\mathbb{N}$ is
\begin{equation}
    Q_n = \min \bigg\{ \min_{m \leq n} P_{m, m}, \; \inf_{m > n} P_{n, m} \bigg\}.
\end{equation}
\end{theorem}

\begin{proof}
Introduce the set $\mathcal{A}$ as the set of all probability distributions on $\mathbb{N}$ and let
\begin{equation}
    \mathcal{F} = \bigg\{ f_k: m \in \mathbb{N} \mapsto \delta_{k, m} \mid k \in \mathbb{N} \bigg\},
\end{equation}
where $\delta_{k, m}$ is the Kronecker delta. Any $p \in \mathcal{A}$ can be written as a linear combination of elements in $\mathcal{F}$ so it holds that $\bar{\mathcal{F}} = \mathcal{A}$. Using corollary \ref{the:equal-extensions} yields
\begin{equation}
    Q_n = Q_n(\mathcal{A}) = Q_n(\mathcal{F}).
\end{equation}
Noticing that the infimum on the set of Kronecker deltas can be written
\begin{equation}
    Q_n(\mathcal{F}) = \min \bigg\{ \min_{m \leq n} P_{m, m}, \; \inf_{m > n} P_{n, m} \bigg\},
\end{equation}
completes the proof.
\end{proof}

\end{document}